\documentclass[sigconf,natbib,screen=true]{acmart}

\PassOptionsToPackage{dvipsnames}{xcolor}

\usepackage{dsfont}
\usepackage{acronym}
\usepackage[noend]{algorithmic}
\usepackage[ruled,vlined,linesnumbered]{algorithm2e}
\usepackage{bbm}
\usepackage{beramono}
\usepackage{bm}
\usepackage{booktabs}
\usepackage[skip=0pt]{caption}
\usepackage{cleveref}
\usepackage[T1]{fontenc}
\usepackage{graphicx}
\usepackage{hyperref}
\usepackage{listings}
\usepackage{multirow}
\usepackage{subcaption}
\usepackage{tabularx}
\usepackage[dvipsnames]{xcolor}
\usepackage[inline]{enumitem}
\usepackage{numprint}
\usepackage[normalem]{ulem}

\usepackage{mathrsfs}
\usepackage{mathtools}

\usepackage{tikz}
\usetikzlibrary{positioning, fadings,patterns,fit}
\usetikzlibrary{arrows}

\makeatletter
\newif\ifAC@uppercase@first%
\def\Aclp#1{\AC@uppercase@firsttrue\aclp{#1}\AC@uppercase@firstfalse}%
\def\AC@aclp#1{%
  \ifcsname fn@#1@PL\endcsname%
    \ifAC@uppercase@first%
      \expandafter\expandafter\expandafter\MakeUppercase\csname fn@#1@PL\endcsname%
    \else%
      \csname fn@#1@PL\endcsname%
    \fi%
  \else%
    \AC@acl{#1}s%
  \fi%
}%
\def\Acp#1{\AC@uppercase@firsttrue\acp{#1}\AC@uppercase@firstfalse}%
\def\AC@acp#1{%
  \ifcsname fn@#1@PL\endcsname%
    \ifAC@uppercase@first%
      \expandafter\expandafter\expandafter\MakeUppercase\csname fn@#1@PL\endcsname%
    \else%
      \csname fn@#1@PL\endcsname%
    \fi%
  \else%
    \AC@ac{#1}s%
  \fi%
}%
\def\Acfp#1{\AC@uppercase@firsttrue\acfp{#1}\AC@uppercase@firstfalse}%
\def\AC@acfp#1{%
  \ifcsname fn@#1@PL\endcsname%
    \ifAC@uppercase@first%
      \expandafter\expandafter\expandafter\MakeUppercase\csname fn@#1@PL\endcsname%
    \else%
      \csname fn@#1@PL\endcsname%
    \fi%
  \else%
    \AC@acf{#1}s%
  \fi%
}%
\def\Acsp#1{\AC@uppercase@firsttrue\acsp{#1}\AC@uppercase@firstfalse}%
\def\AC@acsp#1{%
  \ifcsname fn@#1@PL\endcsname%
    \ifAC@uppercase@first%
      \expandafter\expandafter\expandafter\MakeUppercase\csname fn@#1@PL\endcsname%
    \else%
      \csname fn@#1@PL\endcsname%
    \fi%
  \else%
    \AC@acs{#1}s%
  \fi%
}%
\edef\AC@uppercase@write{\string\ifAC@uppercase@first\string\expandafter\string\MakeUppercase\string\fi\space}%
\def\AC@acrodef#1[#2]#3{%
  \@bsphack%
  \protected@write\@auxout{}{%
    \string\newacro{#1}[#2]{\AC@uppercase@write #3}%
  }\@esphack%
}%
\def\Acl#1{\AC@uppercase@firsttrue\acl{#1}\AC@uppercase@firstfalse}
\def\Acf#1{\AC@uppercase@firsttrue\acf{#1}\AC@uppercase@firstfalse}
\def\Acfi#1{\AC@uppercase@firsttrue\acfi{#1}\AC@uppercase@firstfalse}
\def\Ac#1{\AC@uppercase@firsttrue\ac{#1}\AC@uppercase@firstfalse}
\def\Acs#1{\AC@uppercase@firsttrue\acs{#1}\AC@uppercase@firstfalse}
\robustify\Aclp
\robustify\Acfp
\robustify\Acp
\robustify\Acsp
\robustify\Acl
\robustify\Acf
\robustify\Acfi
\robustify\Ac
\robustify\Acs
\makeatother

\def\AC@@acro#1[#2]#3{%
  \ifAC@nolist%
  \else%
  \ifAC@printonlyused%
    \expandafter\ifx\csname acused@#1\endcsname\AC@used%
       \item[\protect\AC@hypertarget{#1}{\acsfont{#2}}] #3%
          \ifAC@withpage%
            \expandafter\ifx\csname r@acro:#1\endcsname\relax%
               \PackageInfo{acronym}{%
                 Acronym #1 used in text but not spelled out in
                 full in text}%
            \else%
               \dotfill\pageref{acro:#1}%
            \fi\\%
          \fi%
    \fi%
 \else%
    \item[\protect\AC@hypertarget{#1}{\acsfont{#2}}] #3%
 \fi%
 \fi%
 \begingroup
    \def\acroextra##1{}%
    \@bsphack
    \protected@write\@auxout{}%
       {\string\newacro{#1}[\string\AC@hyperlink{#1}{#2}]{\AC@uppercase@write #3}}%
    \@esphack
  \endgroup}

\definecolor{ali}{RGB}{0, 150, 0}
\definecolor{mdr}{RGB}{255, 127, 0}

\newcommand{\ali}[1]{\textcolor{ali}{\bf Ali: #1}}

\newcommand{\mExpect}[2]{\ensuremath{\mathbb{E}_{#1}\left[#2\right]}}

\acrodef{MBC}{mixture-based correction}
\newcommand{\OurMethod}{\ac{MBC}}

\acrodef{IR}{information retrieval}
\acrodef{EM}{Expectation-Maximization}
\acrodef{rbEM}{regression-based EM}
\acrodef{LTR}{learning to rank}
\acrodef{SERP}{Search Engine Results Page}
\acrodef{CLTR}{counterfactual learning to rank}
\acrodef{OLTR}{Online Learning to Rank}
\acrodef{IPS}{inverse propensity scoring}
\acrodef{PBM}{Position-Based Model}
\acrodef{CM}{Cascade Model}
\acrodef{ARP}{Average Relevance Position}
\acrodef{DCG}{Discounted Cumulative Gain}
\acrodef{DLA}{Dual Learning Algorithm}
\acrodef{MNAR}{Missing-Not-At-Random}
\acrodef{mPMM}{modified Poisson Mixture Model}
\acrodef{DCM}{Dependent Click Model}
\acrodef{UBM}{User Browsing Model}
\acrodef{CCM}{Click Chain Model}
\acrodef{PMM}{Poisson Mixture Model}
\acrodef{GMM}{Gaussian Mixture Model}
\acrodef{BMM}{Binomial Mixture Model}
\acrodef{CTR}{click-through rate}
\acrodef{AC}{affine correction}
\acrodef{i.i.d.}{independent and identically distributed}
\acrodef{CLT}{Central Limit Theorem}
\acrodef{NDCG}{Normalized Discounted Cumulative Gain}
\acrodef{MAP}{Mean Average Precision}

\theoremstyle{definition}

\theoremstyle{definition}

\newtheorem{remark}{Remark}[]

\allowdisplaybreaks

\makeatletter
\renewcommand\paragraph{\@startsection{paragraph}{4}{\parindent}%
  {-.1\baselineskip \@plus -1\p@ \@minus -.1\p@}%
  {-2\p@}%
  {\ACM@NRadjust{\@parfont\@adddotafter}}}
\makeatother  
  
\copyrightyear{2021}
\acmYear{2021}
\setcopyright{acmcopyright}\acmConference[CIKM '21]{Proceedings of the 30th ACM International Conference on Information and Knowledge Management}{November 1--5, 2021}{Virtual Event, QLD, Australia}
\acmBooktitle{Proceedings of the 30th ACM International Conference on Information and Knowledge Management (CIKM '21), November 1--5, 2021, Virtual Event, QLD, Australia}
\acmPrice{15.00}
\acmDOI{10.1145/3459637.3482275}
\acmISBN{978-1-4503-8446-9/21/11}

\author{Ali Vardasbi}
\affiliation{%
  \institution{University of Amsterdam}
  \city{Amsterdam}
  \country{The Netherlands}
}
\email{a.vardasbi@uva.nl}

\author{Maarten de Rijke}
\affiliation{%
  \institution{University of Amsterdam}
  \city{Amsterdam}
  \country{The Netherlands}
}
\email{m.derijke@uva.nl}

\author{Ilya Markov}
\affiliation{%
  \institution{University of Amsterdam}
  \city{Amsterdam}
  \country{The Netherlands}
}  
\email{i.markov@uva.nl}

\acmSubmissionID{522}

\title[Mixture-Based Correction for Position and Trust Bias]{Mixture-Based Correction for Position and Trust Bias in Counterfactual Learning to Rank}

\allowdisplaybreaks

\begin{document}

\begin{abstract}
In \ac{CLTR} user interactions are used as a source of supervision.
Since user interactions come with bias, an important focus of research in this field lies in developing methods to correct for the bias of interactions.
\Ac{IPS} is a popular method suitable for correcting position bias.
\Ac{AC} is a generalization of \ac{IPS} that corrects for position bias and trust bias.
\Ac{IPS} and \ac{AC} provably remove bias, conditioned on an accurate estimation of the bias parameters.
Estimating the bias parameters, in turn, requires an accurate estimation of the relevance probabilities.
This cyclic dependency introduces practical limitations in terms of sensitivity, convergence and efficiency. 

We propose a new correction method for position and trust bias in \ac{CLTR} in which, unlike the existing methods, the correction does not rely on relevance estimation.
Our proposed method, \OurMethod{}, is based on the assumption that the distribution of the \aclp{CTR} over the items being ranked is a mixture of two distributions: the distribution of \aclp{CTR} for relevant items and the distribution of \aclp{CTR} for non-relevant items.
We prove that our method is unbiased.
The validity of our proof is not conditioned on accurate bias parameter estimation.
Our experiments show that~\OurMethod{}, when used in different bias settings and accompanied by different \acl{LTR} algorithms, outperforms \ac{AC}, the state-of-the-art method for correcting position and trust bias, in some settings, while performing on par in other settings.
Furthermore,~\OurMethod{} is orders of magnitude more efficient than \ac{AC} in terms of the training time.
\end{abstract}


\maketitle

\acresetall


\section{Introduction}
\label{sec:intro} 
\Ac{LTR} is the practice of using supervision to train a ranking function.
Traditional \ac{LTR} methods use explicit relevance labels produced by human annotators~\citep{liu2009learning}.
In contrast, \ac{CLTR} uses historical interactions, such as clicks, as labels.
Unlike costly manual labels, clicks are available in large amounts for almost no additional cost.
The downside of using clicks as relevance labels, however, is \emph{bias}.
Clicks suffer from different types of bias, such as position bias, selection bias, trust bias, etc.~\citep{joachims2005accurately}
As a result of bias, the probability of a click is not the same as the probability of relevance.
Thus, in order to use clicks as relevance labels, we should first \emph{correct} for the bias~\citep{wang2016learning,joachims2017unbiased}.

A number of techniques have been proposed to debias clicks and to estimate the probability of relevance based on the probability of clicks.
A well-known method is \acfi{IPS}~\citep{joachims2017unbiased,wang2016learning}, which corrects for the position bias in clicks.
\Ac{IPS} relies on the \emph{examination hypothesis}, i.e., an item is clicked if it is examined and perceived to be relevant by a user.
As the name suggests, in \ac{IPS} clicks are re-weighed by the inverse of the examination probability, a.k.a.\ propensity.
\Ac{IPS} is proved to be unbiased when the clicks suffer from position bias~\citep{joachims2017unbiased}.
\Acfi{AC} generalizes \ac{IPS} to also correct for trust bias~\citep{vardasbi2020when}.
\Ac{AC} has been proved to be unbiased when the clicks suffer from both position and trust bias.
The proofs of the unbiasedness of \ac{IPS} and \ac{AC} depend on knowledge of the bias parameters.
Accurately estimating the bias parameters, in turn, depends on obtaining accurate relevance estimations, which is as hard as the \ac{LTR} problem itself.
In the literature, this cyclic dependency is solved by a \ac{rbEM} algorithm that simultaneously learns the ranker as well as the bias parameters~\citep{wang2018position,agarwal2019addressing,vardasbi2020when}.
We argue that integration of a regression function into the standard \ac{EM} leads to a number of practical limitations in terms of
\begin{enumerate*}[label=(\arabic*)]
	\item sensitivity to the regression function,
	\item a lack of guarantees that \ac{EM} converges to a zero gradient, and 
	\item low efficiency of the algorithm.
\end{enumerate*}

We break the curse of cyclic dependency 
by proposing a novel debiasing method, \acfi{MBC}.
Inspired by the idea of score distributions~\citep{arampatzis-2009-score}, we assume that the probability of seeing a specific \ac{CTR} for an item at a position in a ranking is a mixture of \ac{CTR} probabilities for relevant and non-relevant items appearing on that position.
More specifically, we assume that an item is clicked if, and only if, one of the following two disjoint events occurs:
\begin{enumerate*}
	\item a user examines the item and the item is actually relevant (i.e., this is a click on a relevant item), or
	\item the user examines the item, the item is not relevant, but the user clicks on it anyway due a certain bias, e.g., trust in the search engine~\citep{agarwal2019addressing} or visual attractiveness of the item~\cite{chen2012beyond} (i.e., this is a click on a non-relevant item).
\end{enumerate*}
For each position in a ranking, \OurMethod{} estimates the distribution of \acp{CTR} for these two events and calculates the full distribution of \acp{CTR} on that position as their mixture.
Then, the probability of relevance for a given item is calculated by Bayes' rule, as the posterior probability of relevance, given the observed \ac{CTR} over that item.
Finally, the estimated probabilities of relevance are used as labels in a standard \ac{LTR} algorithm.

We prove that~\OurMethod{} gives an unbiased estimator of the probability of relevance, without any prior knowledge of the bias parameter values.
This is a step forward, as \ac{IPS} and \ac{AC} \emph{do} rely on prior knowledge of the bias parameter values to be unbiased.
Below, we show theoretically how inaccurate bias parameters prevent \ac{AC} from completely removing bias.

We confirm our theoretical advances with a set of semi-syn\-thetic experiments. 
We show that the ranking performance of different \ac{LTR} algorithms, trained on the relevance estimates of~\OurMethod{}, always converges to the ranking performance where the true relevance labels are available.
We also compare~\OurMethod{} with the state-of-the-art correction method for position and trust bias, i.e., \ac{AC}.
We compare them by training \ac{LTR} algorithms over~\OurMethod{}'s and \ac{AC}'s respective corrected outputs.
We show that in several cases~\OurMethod{} outperforms \ac{AC} by filling the gap between \ac{AC}'s ranking performance and the true relevance case.
Finally, since both~\OurMethod{} and \ac{AC} depend to the assumption of a click model to infer and correct for the bias, we conduct robustness experiments in terms of click model mismatch.
Specifically, we show that when clicks adhere to the \ac{DCM} or the \ac{UBM}, but a different click model, such as the \ac{PBM}, is assumed by the correction methods,~\OurMethod{} is more robust, i.e., its ranking performance is affected less compared to \ac{AC}.

In summary, the contributions of the paper are:
\begin{enumerate}[leftmargin=*,nosep]
	\item We propose a new debiasing method,~\acf{MBC}, for correcting position and trust bias, and prove its unbiasedness.
	Our proof is stronger than the unbiasedness proofs for existing methods, as it does not rely on the assumption that the bias parameters are known.
	\item We show experimentally that, when used with \ac{LTR} methods, \OurMethod{} outperforms \ac{AC}, the state-of-the-art correction method for position and trust bias, in several settings, while having similar performance in other settings.
	\item We show that~\OurMethod{} is orders of magnitude faster than \ac{AC}, in terms of the training time.
	\item 
	We provide experimental evidence that~\OurMethod{} is more robust to click model mismatch compared to \ac{AC}.
\end{enumerate}


\section{Background}
\label{sec:AC}

The majority of prior work on unbiased \ac{LTR}, tries to correct for the mismatch between the distribution of clicks and relevance probability due to bias.
Bias in clicks means that not all relevant items have the same a priori chance of being clicked.
E.g., position bias means that relevant items at the top of a result list usually absorb more clicks than lower ranked relevant items~\citep{joachims2005accurately};
and trust bias means that users trust a search engine and click on higher ranked non-relevant items more than lower ranked items~\citep{agarwal2019addressing}.
Usually, these types of bias are modeled with the help of click models~\citep{chuklin2015click}.

Below, we review existing methods for correcting position and trust bias.
After discussing \ac{AC} as the state-of-the-art correction method for position and trust bias, we analyze its relevance estimation error and show how inaccurate bias parameters cause \ac{AC} to remain biased.
We also discuss other work related to this paper.

\subsection{A review of \ac{AC}}

\citet{agarwal2019addressing} notice that \ac{IPS} is still biased when there is also trust bias.
In the presence of trust bias, the click probability should be written as follows (for brevity we drop the $(\cdot \mid q,d,k)$ conditions from all the probabilities, where $q$ and $d$ represent the query and item and $k$ is the item position in the results list.):
\begin{equation}
    \label{eq:AC:trustbias}
    \begin{split}
        P(C=1)  =~&P(E=1) \cdot P(R=1) \cdot P(C=1 \mid R=1,E=1) \\
        ~&+ P(E=1) \cdot P(R=0) \cdot P(C=1 \mid R=0,E=1)\\
        =~&\alpha P(R=1) + \beta.
    \end{split}
\end{equation}

\noindent%
where $C$, $E$ and $R$ indicate click, examination and relevance binary indicators.
\citet{vardasbi2020when} prove that the following correction gives an unbiased estimate of the relevance in this situation:
\begin{equation}
    \label{eq:AC:AC}
    \hat{r}_{q,d} = \frac{c_{q,d,k}-\beta_{q,d,k}}{\alpha_{q,d,k}}.
\end{equation}
\noindent
where $c_{q,d,k}$ is the click over document $d$ of query $q$ at position $k$; $\alpha$ and $\beta$ are bias parameters (Eq.~\ref{eq:AC:trustbias}); and $\hat{r}_{q,d}$ is the estimated relevance of $d$.

\subsection{\Acl{rbEM}}
\Ac{IPS} and \ac{AC} are unbiased only if the value of the bias parameters such as $\alpha$ and $\beta$ are known, or accurately estimated.
Since the standard \ac{EM} requires the availability of multiple sessions of the same query with different items ordering to work properly~\citep{DBN2009},~\citet{wang2018position} proposed to use \acfi{rbEM} to solve the sparsity problem.
In the \ac{rbEM}, the $P(R=1 \mid q,d)$ values obtained in the M-step are first used to fit a regression function, and then, the output of this regression function is used in the next E-step.
Though \ac{rbEM} leads to good results in the \ac{CLTR} framework with \ac{IPS} and \ac{AC}~\citep{agarwal2019addressing,vardasbi2020when,wang2018position}, in this paper we argue that the integration of a regression function into the \ac{EM} leads to multiple practical limitations, as we explain in Appendix~\ref{app:limitations}.

\subsection{Error analysis of \ac{AC}}
\label{sec:existing:error}
Let us denote the true relevance of document $d$ to query $q$ by $r_{q,d}$, and the probability $P(r_{q,d}=1)$ by $\gamma_{q,d}$.
In the presence of trust bias, according to Eq.~\eqref{eq:AC:trustbias}, we have:
\begin{equation}
    \mExpect{c}{c_{q,d,k}} = \alpha_{q,d,k}\cdot\gamma_{q,d} + \beta_{q,d,k}.
\end{equation}
In what follows we drop the subscripts for brevity.

Supposing $\alpha'$ and $\beta'$ are estimates of $\alpha$ and $\beta$ obtained from the \ac{rbEM} algorithm, \ac{AC} estimates $r$ as follows:
\begin{equation}
    \hat{r} = \frac{c - \beta'}{\alpha'}.
    \label{eq:existing:affine}
\end{equation}
In order to have an unbiased estimator, we need to ensure that $\mExpect{c}{\hat{r}} = \mExpect{r}{r}$.
Therefore, we are interested in 
    $e = \left| \hat{r} - r \right|$.
Using Eq.~\eqref{eq:AC:trustbias}, we can calculate the expectation as follows:
\begin{equation}
    \begin{split}
    \mExpect{c,r}{e} &= \gamma \left|1 - \frac{\alpha + \beta - \beta'}{\alpha'}\right| +
    \left(1 - \gamma\right) \left|\frac{\beta - \beta'}{\alpha'}\right| \\
    &=\gamma \frac{|\Delta \alpha + \Delta \beta| - |\Delta \beta|}{\alpha'} + \frac{|\Delta \beta|}{\alpha'},
    \label{eq:existing:mu}
    \end{split}
\end{equation}
where $\Delta \beta = \beta' - \beta$ and $\Delta \alpha = \alpha' - \alpha$ are estimation errors of the bias parameters.
It is important to have the average error converge to zero as the number of sessions associated with the query grows.
Eq.~\eqref{eq:existing:mu} shows that only when $\Delta \alpha \rightarrow 0$ and  $\Delta \beta \rightarrow 0$, i.e., when the parameter estimations obtained from \ac{rbEM} are accurate, this is the case.
In summary, inaccurate bias parameters estimations cause the \ac{AC} method to remain biased, even with infinitely many training sessions.
\subsection{Mixture of distributions}
Our assumption of having a mixture of relevant and non-relevant distributions is related to score distribution models~\citep{arampatzis-2009-score}.
But~\OurMethod{} differs in two important ways.
First, unlike scores, clicks are feedback and constitute a source of supervision. 
This makes \ac{CLTR} used with \OurMethod{} a supervised learning algorithm as opposed to the unsupervised algorithms based on score distribution models.
Second, a score distribution model is built over the corresponding list of items for a single query, whereas our mixture model is built over the items with the same examination probability for different queries.
In this sense, our model has a global view over all queries, while the score distribution models have a local view over one query.

\subsection{Other related work}
\label{sec:relatedworks}
Instead of using \ac{rbEM}, \citet{ai2018unbiased} propose to use the \ac{DLA}.
\ac{DLA} is shown to be effective in estimating the bias parameters and leading to an unbiased \ac{LTR}.
However, it only models the position bias and it is not clear how it can be extended to work with trust bias.

\citet{qin2020attribute} use \ac{rbEM} to estimate the attribute-based propensity, which considers different platforms and feedback sources in addition to the items positions.
\citet{dia2020urank} define a utility based on the click probability and propose to optimize that utility directly, instead of optimizing retrieval metrics with the hope that they may indirectly improve the \ac{CTR}.
Their approach solves the problem of bias parameter estimation by directly learning a position-aware click model from user interactions.

What our proposed method, \ac{MBC}, contributes on top of the related work discussed above is that it breaks the cyclic dependency between the bias parameters and relevance probability.
This means that in~\OurMethod{} the bias parameters are inferred only using the user interactions, without any direct reliance on the relevance probabilities.
In other words, in existing methods relevance estimation is unbiased if the bias parameters are accurately set, and the bias parameters can be set accurately if the relevance estimation is precise.
In contrast, with~\OurMethod{}, the correction and bias parameter estimation are performed at the same time, without any reliance on relevance estimation.
This enables us to avoid the use of regression functions inside the \ac{EM} algorithm, which is shown in Appendix~\ref{app:em} to have practical limitations.

Finally, in \citep{chandar2018estimating,vardasbi2020cascade} it is argued that a mismatch between the actual model of clicks and the assumed click model for correction hurts ranking performance of the correction methods.
We address this issue by showing that~\OurMethod{} is robust to click model mismatch, specifically, when actual clicks adhere to \ac{DCM} or \ac{UBM} while~\OurMethod{} assumes \ac{PBM} for correction.


\section{Mixture-based Correction}
\label{sec:method}
In this section, we explain our~\acf{MBC} method and prove that
it gives an unbiased estimate of relevance.

\subsection{Method}
\label{sec:method:method}
Our idea is to infer the relevance of an item to a user's query based on the observed \ac{CTR} for that query-item pair.
Similarly to previous work on \ac{CLTR}, we assume that user clicks on search results follow the examination hypothesis~\citep{joachims2017unbiased,ai2018unbiased,agarwal2019addressing,wang2018position,oosterhuis2020topkrankings,vardasbi2020when},
that is, a click on an item (or, consequently, the \ac{CTR} for that item) is affected only by how likely the item is to be examined and perceived relevant by a user.
So there are two components, namely, examination and relevance of an item, that contribute to a click on the item.
Our goal is to estimate the relevance component.

To rule out the examination component, we consider items with the same examination probability $P(E=1)$.
In practice, $P(E=1)$ is not known and one needs a click model to decide which items have the same $P(E=1)$.
We will address this issue later in this section.
For now, assume that we know which items have the same $P(E=1)$.

For a set of items with the same examination probability $P(E=1)$ there is a certain distribution of \ac{CTR}s, $P(CTR = x)$.
Assuming binary relevance,\footnote{Graded relevance can be considered as the probability $P(R=1)$. See Sec.~\ref{sec:experiments} for further discussions.} this distribution has two parts:
one for relevant items and one for non-relevant items.
So $P(CTR = x)$ can be seen as a mixture of two separate distributions: the distribution $P(CTR = x \mid R = 1)$ of CTRs of relevant items and the distribution $P(CTR = x \mid R = 0)$ of \ac{CTR}s of non-relevant items. 
Formally:
\begin{align}
    \label{eq:mixture}
    P(CTR = x) =&\ P(R = 1) \cdot P(CTR = x  \mid R = 1) \nonumber \\
    &{}+\ P(R = 0) \cdot P(CTR = x  \mid R = 0).
\end{align}
\noindent
Now, we can reach our goal and calculate the probability of relevance based on the observed \ac{CTR} using Bayes' rule: 
\begin{equation}
    \label{eq:method:relprob}
    P(R = 1 \mid CTR = x) = \frac{P(R = 1) \cdot P(CTR = x  \mid R = 1)}{P(CTR = x)}.
\end{equation}
These relevance probabilities can ultimately be used for \ac{CLTR}.
Algorithm~\ref{alg:mbc} summarizes the above steps of~\OurMethod{}.

For our~\OurMethod{} method to work, it remains to discuss two things:

\begin{algorithm}
    \caption{\Acl{MBC}}
    \label{alg:mbc}
    \SetAlgoLined
    \KwIn{\Aclp{CTR}}
    \KwOut{Estimates of relevance probability}
    Constitute sets of items with the same $P(E=1)$\;
    \ForAll{set of items with the same $P(E=1)$}{
        Fit a 2-component mixture model to Eq.~\eqref{eq:mixture}\;
        Output the relevance of items according to Eq.~\eqref{eq:method:relprob}\;
    }
\end{algorithm}

\noindent%
\begin{enumerate*}
    \item How to estimate the mixture in Eq.~\eqref{eq:mixture} (Line \numprint{3} in Alg.~\ref{alg:mbc}); and
    \item How to get items with the same examination probability (Line \numprint{1} in Alg.~\ref{alg:mbc}).
\end{enumerate*}
This is what we turn to next.

\paragraph{Mixture estimation}
To infer distributions $P(CTR = x  \mid R = i)$ and priors $P(R=i)$ for $i\in\{0,1\}$ in Eq.~\eqref{eq:mixture}, parametric approaches can be used.
For example, we can assume a Gaussian or a binomial mixture model.
As is common with parametric mixture models, we use standard~\ac{EM} to learn the above distributions and priors~\citep{mclachlan1988mixture}.

The limitation of the parametric approach to estimating mixtures is that it depends on the choice of the underlying parametric distribution (e.g., Gaussian, binomial, etc).
However, in Sec.~\ref{sec:method:error}, we show that as long as there are enough clicks,
the choice of this distribution is not essential to our method.
Also, in our experiments in Sec.~\ref{sec:results:mixture}, we compare Gaussian and binomial mixture distributions and show that, in practice, both converge to the same performance.

\paragraph{Items with the same examination}
The~\OurMethod{} method works on sets of items with the same examination probability.
Note that~\OurMethod{} does not require the exact values of examination probabilities,
it only requires to know which items have the same examination.
To detect the desired sets of items, we propose to use click models~\cite{chuklin2015click}.

For instance, we can assume that clicks adhere to the \ac{PBM} as is common in \ac{CLTR}~\citep{joachims2017unbiased,ai2018unbiased,agarwal2019addressing,wang2018position,oosterhuis2020topkrankings,vardasbi2020when}.
In \ac{PBM}, \emph{all items at the same position in a results list} have the same examination probability.
And that is all we need from the assumed click model.
For cascade models, such as the \ac{DCM}~\citep{DCM2009} and the \ac{UBM}~\citep{UBM2008},
\emph{all items at the same position} and with \emph{the same pattern of relevant items above them} have the same examination probability.
In the case of \ac{DCM} and \ac{UBM}, the desired set of items can be detected recursively:
to collect items with the same examination probability at position $k$,
we first use~\OurMethod{} to detect all relevant items at positions above $k$, and then we group items with the same pattern of relevant items above $k$.
Note that this recursive process is not cyclic: the grouping of clicks at rank $k$ depends on the reliability of the model at rank $k-1$.
Again, no parameter estimation for the assumed click models is required.

\begin{remark}
\label{rem:almostthesame}
    Detecting items with the same examination may seem a bottleneck in real world scenarios.
    For~\OurMethod{} to work properly, the distributions of relevant and non-relevant \ac{CTR}s should be separable.
    This means that the condition of \emph{items with the same examination} can be loosened by \emph{items with \textbf{almost} the same examination}.
    To elaborate, suppose a set of items are considered whose examination probabilities are either $\theta$ or $\theta'$.
    Each examination probability leads to one distribution for relevant \ac{CTR}s, and one for non-relevant \ac{CTR}s: there are two relevant \ac{CTR} distributions and two non-relevant \ac{CTR} distributions in the set.
    The relevant (non-relevant) \ac{CTR} distribution of all the items in the set is itself a mixture of two distributions corresponding to $\theta$ and $\theta'$.
    As long as these two relevant/non-relevant distributions of all items are separable,~\OurMethod{} would be able to distinguish relevant items from non-relevants, and the \ac{LTR} remains unbiased.
    This argument is easily extended to cases with more than two different examination probabilities.
    We will discuss this remark more in Sec.~\ref{sec:results:cmranking_performance} and give a toy example for further clarifications.
\end{remark}

\noindent%
In the remainder of this section, we prove that the~\OurMethod{} inferred relevance labels are unbiased estimations of the true relevance labels.
Consequently, a \ac{LTR} algorithm trained on these corrected values will be an unbiased \ac{LTR}.
\subsection{Unbiasedness of~\OurMethod}
\label{sec:method:error}
In this paper, similar to most previous work on online \ac{LTR} and \ac{CLTR}, we assume that clicks on different sessions are independent~\citep{chuklin2015click, zoghi2017online, oosterhuis2018differentiable,joachims2017unbiased,ai2018unbiased,agarwal2019addressing,wang2018position,oosterhuis2020topkrankings,jagerman2019comparison,vardasbi2020when, vardasbi2020cascade}.
As discussed in Sec.~\ref{sec:method:method}, we consider the set of items with almost the same examination probability and fit a mixture model for each such set.
We assume that clicking on a relevant item is a random variable with mean $\mu_1$ and variance $\sigma^2_1$.
Similarly, clicking on a non-relevant item has mean $\mu_0$ and variance $\sigma^2_0$.
For a unique query, repeated in $n$ sessions, assume the clicks over the item $x$ are given by $\{c^{(x)}_1,c^{(x)}_2,\ldots,c^{(x)}_n\}$.
We prove that the \ac{CTR} defined as
\begin{equation}
    v_x=\frac{c^{(x)}_1+c^{(x)}_2+\cdots+c^{(x)}_n}{n},
\end{equation}
can be used to estimate the relevance of $x$, given that $n$ is large enough.

\begin{theorem}
Assuming independent clicks in different sessions, the clustering of the \ac{CTR} signal into relevant and non-relevant items converges in probability to the true relevance of the items.
\label{thm:MBC}
\end{theorem}

\begin{proof}
Based on the assumptions, the $c^{(x)}_i$'s constitute a sequence of \ac{i.i.d.} random variables.
According to the \ac{CLT}, as $n$ grows, $v_x$ will converge in probability to $\mExpect{c}{c^{(x)}_i}$, which is either $\mu_0$ or $\mu_1$.
We are interested in the case where a full recovery of mixture models is possible, i.e., one can fully separate the distribution of relevant \acp{CTR} from the distribution of non-relevant \acp{CTR}.
The variance of $v_x$ diminishes linearly by $n$.
Consequently, given any threshold value for the variance for which a full recovery is possible, there is a sufficiently large $n$ that leads to the given threshold value.
This means that there exists an $n$ such that a full recovery of the mixture components is possible.
\end{proof}

\noindent%
We used the \acl{CLT} for the proof of Theorem~\ref{thm:MBC}, which does not rely on any specific distribution.
In order to get a better understanding of what constitutes a sufficiently large $n$ for full recovery, we discuss the special case of a Gaussian mixture.
There is a rich literature on the analysis of the recoverability of Gaussian mixture models \citep[see, e.g.,][]{chen2020cutoff}.
In~\citep{lu2016statistical} a simple condition is given for \emph{almost full recovery} of Gaussian mixtures, which translates to the setting of this paper as follows:
\begin{equation}
    n = \Omega\left(\left(\frac{\max(\sigma_0, \sigma_1)}{\mu_1 - \mu_0}\right)^2\right).
    \label{eq:method:afullrecovery}
\end{equation}
\noindent
This means that for any $\mu_i$ and $\sigma_i$ values, if we increase the number of sessions $n$ so that Eq.~\eqref{eq:method:afullrecovery} holds, the \ac{CTR} of relevant and non-relevant items can be almost fully separated.
In our experiments, where the parameters are set based on previous real-world studies such as~\citep{agarwal2019addressing}, Eq.~\eqref{eq:method:afullrecovery} becomes $n=\Omega(1)$ and we find $n\geq 15$ to be a suitable value based on the convergence of ranking performance.

\section{Experimental setup}
\label{sec:experiments} 
In this section, we discuss the experimental setup used to demonstrate the effectiveness of our proposed~\OurMethod{} method.
Following previous \ac{CLTR} studies~\citep{joachims2017unbiased, ai2018unbiased, jagerman2019comparison,agarwal2019addressing, oosterhuis2020topkrankings, vardasbi2020when}, we measure the effectiveness of click debiasing by the ranking performance of \ac{LTR} when used on top of debiasing methods.
Our experiments are performed on two publicly available \ac{LTR} datasets with query-document features and relevance labels, while clicks are simulated.

\subsection{Datasets}
As a regular choice in \ac{LTR} research~\citep{joachims2017unbiased, ai2018unbiased, oosterhuis2020topkrankings, jagerman2019comparison, vardasbi2020when},
we use two popular \ac{LTR} datasets: Yahoo! Webscope~\citep{Chapelle2011} and MSLR-WEB30k~\citep{qin2013introducing}.
For each query, these datasets contain a list of documents with human-generated 5-level relevance labels.
Yahoo!\ has \numprint{29670} queries, 23.84 documents per query on average, and uses 700-feature vectors to represent query-documents;
MSLR has \numprint{31339} queries, 120.19 documents per query, and uses 136 features.
We use the default provided training, validation and test splits for each dataset.
We only use the first fold of MSLR.

\subsection{Click simulation}
\label{sec:experiment:clicksimulation}
In our experiments, we measure the performance of \ac{LTR} trained on user clicks.
The Yahoo! and MSLR datasets, however, do not contain click information.
Following a long line of previous studies~\citep{joachims2017unbiased, ai2018unbiased, oosterhuis2020topkrankings, jagerman2019comparison, vardasbi2020when,hofmann-2011-balancing},
we simulate clicks as follows.

\paragraph{Production ranker}
First, we simulate a production ranker that is used to create a ranked list of items for each query.
Following~\cite{joachims2017unbiased}, we train a \ac{LTR} method on a very small number of randomly selected queries.
The intuition is to provide a production ranker that is better than a random ranker, but still has room for improvement.
We select 20 random queries from each dataset and use LambdaMART~\citep{burges2010ranknet} to train a production ranker.
To train LambdaMART, we use the LightGBM package (version 3.2.1.99)~\citep{ke2017lightgbm} with the following parameters:
300 trees, 31 leaves and a learning rate of $0.05$.
Training in this way with 20 queries gives us a decent ranker that performs considerably better than a random ranker, while still having room for improvement.
Hence, we can exploit user interactions to improve this production ranker.

The production ranker is then used to rank items for each query.
We cut the list of items at top-$m$, as do real-world search engines.
We report results for $m = 20$.
We also performed experiments with the top-$50$, but since the results showed no significant difference compared to the top-$20$, we do not report them in the paper.

\paragraph{User clicks}
To simulate clicks, we follow previous studies on trust bias~\citep{agarwal2019addressing,vardasbi2020when}.
Given a list of items returned by the production ranker for a query $q$, we first compute the click probability for each item $x$ and position $k$ using Eq.~\eqref{eq:AC:trustbias} and the \ac{PBM} assumption.
(For experiments in Sec.~\ref{sec:results:cmranking_performance} we replace \ac{PBM} with \ac{DCM} and \ac{UBM}.)
Then, for each item $x$ and position $k$ we simulate a click by sampling from this Bernoulli distribution.
Unless stated otherwise, we use $8M$ clicks (with uniformly repeated queries) to train the correction methods.
In our experiments, both~\OurMethod{} and \ac{AC} begin to converge to their final performance after $2M$ clicks.
We choose $8M$ to be on the safe side.
Eq.~\eqref{eq:AC:trustbias} depends on two quantities:
\begin{enumerate*}
    \item relevance probabilities; and
    \item bias parameters.
\end{enumerate*}
We will discuss both bellow.

\paragraph{Relevance probabilities}
Both datasets used in our experiments provide graded relevance labels.
For simulating the clicks in a graded relevance setting, a transformation function is required to change the integer grades into valid relevance probabilities.
We employ the following two strategies, assuming $y\in\{0,\ldots,y_{\max}\}$ is the relevance grade:
\begin{enumerate}[leftmargin=*,nosep]
    \item {\bfseries Binarized}: Following~\cite{joachims2017unbiased,vardasbi2020when} relevance probability can itself be binary:
        $P(R=1 \mid y) = 1 \quad\text{iff}\quad y > \frac{y_{\max}}{2}$.
    \item {\bfseries Graded}: The grades can also be turned into probabilities using a linear transformation~\citep{oosterhuis2021unifying}:
        $P(R=1 \mid y) = \frac{y}{y_{\max}}$.
\end{enumerate}

\paragraph{Bias parameters}
Similarly to~\citep{joachims2017unbiased, ai2018unbiased, oosterhuis2020topkrankings, jagerman2019comparison, vardasbi2020when}, the position bias for a position $k$ is computed as 
    $P(E=1 \mid k) = k^{-\eta}$,
where the parameter $\eta$ controls the severity of the position bias.
Usually, $\eta = 1$ is considered in the \ac{CLTR} literature~\citep{joachims2017unbiased,ai2018unbiased,vardasbi2020when}.
In our experiments we consider $\eta \in \{1, 2\}$.

For the trust bias, we follow~\citep{vardasbi2020when} and for each position $k$ compute the parameters as follows:\footnote{Similar values were reported in the real-world experiments in~\citep{agarwal2019addressing} and used in simulations in~\citep{oosterhuis2021unifying}.}
\begin{align}
    \label{eq:epsilon}
    P(C=1 \mid R=1,E=1, k) & = 1 - \frac{\min(k, 20)+1}{100}
    \\
    P(C=1 \mid R=0,E=1, k) & = \frac{0.65}{\min(k, 10)}.
\end{align}

\subsection{\Acl{LTR}}
We train a \ac{LTR} method over the corrected output of different click debiasing methods.
In order to show the consistency of our results over different \ac{LTR} methods, we use two \ac{LTR} approaches:
\begin{enumerate*}
    \item {\bfseries LambdaMART}~\citep{burges2010ranknet}; and
    \item {\bfseries DNN}, a neural network similar to that of~\citep{vardasbi2020when}.
\end{enumerate*}
The LambdaMART implementation of LightGBM only works with integer labels, but the input of \ac{LTR} in our experiments is the relevance probability, which is non-integer.
To solve this problem, we made some minor modifications to the source code of LightGBM in order to make it work with non-integer inputs as well.
The changed source files are included in the code repository of this paper (see the end of the paper for the link).



\subsection{Baselines}
We compare the results of \OurMethod{} to those of \ac{AC}~\citep{vardasbi2020when}, the state-of-the-art click debiasing method for position and trust bias.
We do not include \ac{IPS} in our comparison, since 
\begin{enumerate*}
    \item it cannot correct for trust bias, and
    \item \ac{AC} is a generalization of \ac{IPS} that leads to the same performance when trust bias is absent~\citep{vardasbi2020when}.
\end{enumerate*}
In summary, we compare the ranking performance of \ac{LTR} trained on debiased clicks of our~\OurMethod{} method, with the following settings:
\begin{enumerate}[leftmargin=*,nosep]
    \item \textbf{Ideal-\ac{AC}}: \ac{AC} with true bias parameters.
    This gives the highest potential of \ac{AC} and is not realistic, since it uses the true values for bias parameters;
    \item \textbf{\ac{AC}}: \ac{LTR} trained on debiased clicks of \ac{AC}, using \ac{rbEM} for propensity estimation (See~\ref{sec:results:eminstability} for the choice of regression function);
    \item \textbf{No correction}: \ac{LTR} trained on clicks without any debiasing;
    \item \textbf{Relevance probabilities}: \ac{LTR} trained on the true relevance probabilities, obtained from the true relevance labels. 
    This depends on the strategy that is used to transform the integer relevance grades into probabilities (see paragraph \emph{Relevance probabilities} in Sec.~\ref{sec:experiment:clicksimulation}).
\end{enumerate}

\noindent%
In our experiments, we use the same \ac{LTR} algorithm for all the methods listed above, so that any differences in ranking performance are solely due to the correction method, and no other factors such as \ac{LTR} performance.


\subsection{Metrics}
We measure the ranking performance of different methods by \ac{NDCG}.
As we are evaluating using graded relevance datasets, other metrics such as MAP cannot be used.
All reported nDCG@10 results are an average of eight independent runs.
We use the Student's t-test to determine significant differences.


\section{Results}
\label{sec:results} 
Our experimental results are centered around three benefits of \OurMethod{} compared to \ac{AC}.
The first -- and most important -- benefit is the ranking performance (Sec.~\ref{sec:results:ranking_performance}) and efficiency (Sec.~\ref{sec:result:efficiency}) improvement.
The second is its robustness to click model mismatch (Sec.~\ref{sec:results:cmranking_performance}).
The third benefit is that it requires almost no hyper-parameter tuning as opposed to \ac{rbEM}.
Since \OurMethod{} is based on a standard mixture model, it is free of the so-called hyper-parameters prevalent in deep learning models.
On the other hand, \ac{rbEM} uses a regression function that has a significant impact on the effectiveness and efficiency of the unbiased \ac{LTR} algorithm (as we show in Appendix~\ref{sec:results:eminstability}).



\subsection{An insider's look into corrected clicks}
\label{sec:results:insiders}
Before proceeding with the main experimental analysis, we compare the shape of the corrected clicks obtained from~\OurMethod{} and \ac{AC}.
This qualitative analysis is insightful in understanding the intrinsic difference between the two methods.
\setlength{\tabcolsep}{0.2em}

{\renewcommand{\arraystretch}{0.3}
\begin{figure}[t]
\centering
\begin{tabular}{l c c}
& \small Position $1$
& \small Position $9$
\\
\begin{tabular}{c}
\rotatebox[origin=lt]{90}{\small Estimated $P(R=1)$}
\end{tabular}
&
\begin{tabular}{c}
    \includegraphics[scale=0.33]{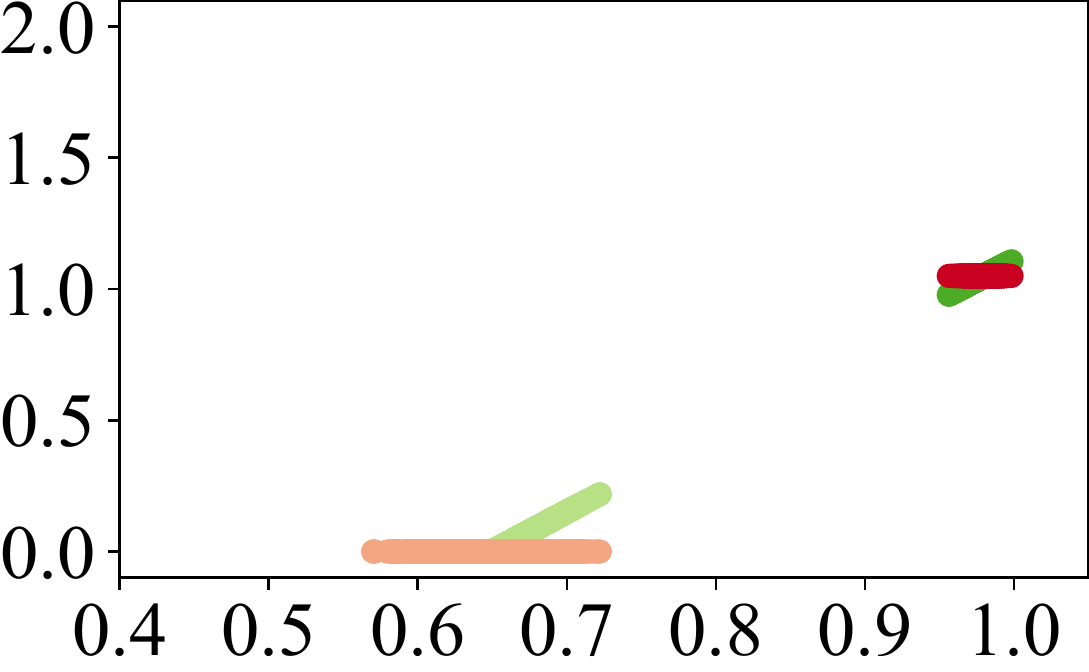} \\
    \small CTR
\end{tabular}
&
\begin{tabular}{c}
    \includegraphics[scale=0.33]{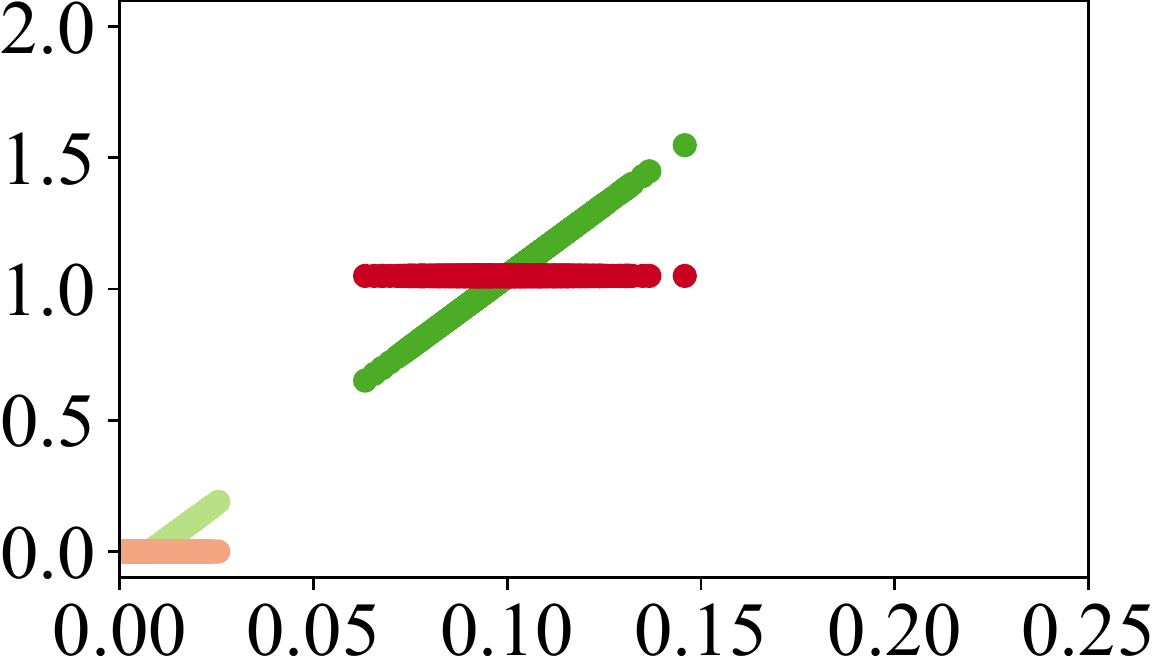} \\
    \small CTR
\end{tabular} 
\\ 
 \multicolumn{3}{c}{
 \includegraphics[scale=.45]{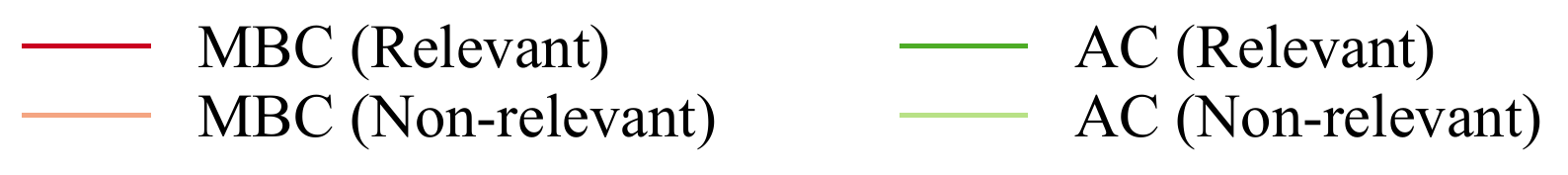}
} 
\end{tabular}
\caption{
Correction comparison of our~\OurMethod{} method with \ac{AC} in inferring the relevance label from the observed CTR.
}
\label{fig:results:insider_look}
\end{figure}
}

Fig.~\ref{fig:results:insider_look} looks into the corrected \ac{CTR} of~\OurMethod{} and \ac{AC}.
In this figure, the \ac{CTR} in two different positions of a ranked list are observed.
We used the binarized relevance probability (Sec.~\ref{sec:experiments}) in this figure for simpler illustration purposes.
We see that with enough sessions per query, the relevant and non-relevant items become completely disjoint.
Consequently,~\OurMethod{} is able to correctly distinguish between the relevant and non-relevant items and infer their relevance label.
The complex transformation of \acp{CTR} into relevance labels in~\OurMethod{} allows for accurate estimations of relevance.
On the other hand, the linearity of~\ac{AC} leads to wider range of relevance labels, scattered around the true values of zero and one.




\begin{table*}[t]
    \caption{NDCG@10 comparison of~\OurMethod{} and \ac{AC} on Yahoo! Webscope and MSLR-WEB30k datasets.
    Superscripts $^*$ and $^\dagger$ indicate significance compared to the other correction method with $p<0.01$ and $p<0.1$, respectively.}
    \label{tab:results:rankingcomp}
    \setlength{\tabcolsep}{0.6em}
    \begin{tabular}{l l cccc| cccc}
        \toprule
        &
        & \multicolumn{4}{c|}{Yahoo! Webscope}
        & \multicolumn{4}{|c}{MSLR-WEB30k}
        \\
        \cmidrule(lr){3-6}\cmidrule(lr){7-10} 
        &
        & \multicolumn{2}{c}{Normal bias ($\eta=1$)} & \multicolumn{2}{c|}{Severe bias  ($\eta=2$)}
        & \multicolumn{2}{|c}{Normal bias ($\eta=1$)} & \multicolumn{2}{c}{Severe bias  ($\eta=2$)}
        \\
        \cmidrule(lr){3-4}\cmidrule(lr){5-6}\cmidrule(lr){7-8}\cmidrule(lr){9-10} 
        &
        & \bf Binarized     & \bf Graded
        & \bf Binarized     & \bf Graded
        & \bf Binarized     & \bf Graded
        & \bf Binarized     & \bf Graded
        \\
       
        \midrule
        \multirow{7}{*}{LambdaMART}
        & \multicolumn{5}{l}{\emph{Trained using clicks}}
        \\
        \cmidrule(lr){2-10}

& No Correction
& 0.692 & 0.691 & 0.689 & 0.689 & 0.400 & 0.402 & 0.396 & 0.397 \\
& AC
& 0.758 & 0.763 & 0.725 & 0.740 & 0.426 & \bfseries 0.477\rlap{$^\dagger$} & 0.337 & 0.431 \\
& \bfseries MBC
& \bfseries 0.760\rlap{$^*$} & \bfseries 0.767\rlap{$^\dagger$} & \bfseries 0.753\rlap{$^*$} & \bfseries 0.748\rlap{$^*$} & \bfseries 0.428 & 0.474 & \bfseries 0.420\rlap{$^*$} & \bfseries 0.454\rlap{$^*$} \\

        \cmidrule(lr){2-10}
        & \multicolumn{5}{l}{\emph{Trained using oracle knowledge}}
        \\
        \cmidrule(lr){2-10}
    
& Ideal-AC
& 0.758 & 0.771 & 0.725 & 0.740 & 0.423 & 0.486 & 0.337 & 0.432 \\
& Rel. Probs
& 0.759 & 0.774 & 0.759 & 0.774 & 0.429 & 0.491 & 0.429 & 0.491 \\
    
        \midrule
        \multirow{7}{*}{DNN}
        & \multicolumn{5}{l}{\emph{Trained using clicks}}
        \\
        \cmidrule(lr){2-10}

& No Correction
& 0.713 & 0.716 & 0.710 & 0.707 & 0.401 & 0.415 & 0.402 & 0.400 \\
& AC
& 0.736 & \bfseries 0.746\rlap{$^*$} & 0.736 & \bfseries 0.746\rlap{$^*$} & 0.404 & \bfseries 0.448\rlap{$^*$} & 0.405 & \bfseries 0.448\rlap{$^*$} \\
& \bfseries MBC
& \bfseries 0.744\rlap{$^*$} & 0.744 & \bfseries 0.741\rlap{$^*$} & 0.740 & \bfseries 0.406 & 0.440 & \bfseries 0.419\rlap{$^*$} & 0.439 \\

        \cmidrule(lr){2-10}
        & \multicolumn{5}{l}{\emph{Trained using oracle knowledge}}
        \\
        \cmidrule(lr){2-10}
    
& Ideal-AC
& 0.742 & 0.749 & 0.737 & 0.746 & 0.409 & 0.451 & 0.405 & 0.446 \\
& Rel. Probs
& 0.743 & 0.750 & 0.743 & 0.750 & 0.416 & 0.453 & 0.416 & 0.453 \\

        \bottomrule
    \end{tabular}
    \end{table*}
\subsection{Ranking performance of \OurMethod{} and \ac{AC}}
\label{sec:results:ranking_performance}
In this section, we try to answer the main research question of this paper:
\begin{enumerate}
    \item[] \emph{How does the~\OurMethod{} method perform compared to \ac{AC} as the state-of-the-art method for correcting position and trust bias?}
\end{enumerate}

\noindent%
In order to determine the effectiveness of methods in debiasing clicks, we measure the ranking performance of a selected \ac{LTR} algorithm trained over the corrected output of~\OurMethod{} and \ac{AC}.

Table~\ref{tab:results:rankingcomp} shows the comparison of~\OurMethod{} and \ac{AC} in terms of NDCG@10.
In this table, we compare the performance for different \ac{LTR} algorithms: LambdaMART and DNN, and different strategies for transforming relevance grades to probabilities: binarized and graded.

The results show that in most cases~\OurMethod{} performs significantly better than AC.
When LambdaMART is used as the \ac{LTR} algorithm,~\OurMethod{} significantly outperforms \ac{AC} on both datasets, both binarized and graded settings and both bias severity cases ($\eta=1$ and $\eta=2$), the only exception being the graded MSLR with normal position bias.
Things are different for the DNN case.
\Ac{AC} significantly outperforms~\OurMethod{} in the graded setting in both datasets.
We also see that, with a single correction method, in some settings LambdaMART performs better than DNN, while in others DNN is the winner.
Considering the \ac{LTR} algorithm as a hyperparameter, i.e. for each correction method in each setting, selecting the \ac{LTR} with the higher ranking performance, we can claim that~\OurMethod{} corrects more effectively than \ac{AC}, as the best performing~\OurMethod{} leads to better results than the best performing \ac{AC}.
For instance, in graded Yahoo! with severe bias, the best performing \ac{AC} is obtained from DNN: $0.746$ vs $0.740$, while the best performing~\OurMethod{} is due to LambdaMART: $0.748$ vs $0.740$.
We see that LambdaMART~\OurMethod{} outperforms DNN \ac{AC} in this case.

As can be seen in Table~\ref{tab:results:rankingcomp}, there is a gap between \ac{AC} (where the bias parameters are estimated using \ac{rbEM}) and Ideal-\ac{AC} (where the bias parameters are assumed to be known by an oracle).
This shows the dependency of \ac{AC} on the accuracy of the bias parameters.
As a reminder, this dependency is one of our motivations to propose the novel~\OurMethod{} method.

With severe position bias ($\eta=2$), as a result of very low examination probabilities, we observe a noticeable drop in the performance of \ac{AC} and Ideal-\ac{AC} in some cases.
Specifically, for the binarized MSLR with LambdaMART, we observe that \ac{AC} performs even worse than the no correction case.
This observation underlines the need for variance reduction techniques for the \ac{AC} method.


\begin{remark}
We see that in some of the binarized cases,~\OurMethod{} performs slightly better than the Rel. Probs case, which may seem counterintuitive.
The reason is that the evaluation is performed on the graded test set for comparability considerations.
As a result, the relevance grades of $\{3,4\}$ as well as $\{0,1,2\}$ are treated the same in the training, but distinguished in evaluation.
This is further observable in comparing the Rel. Probs in the binarized and graded settings: the binarized Rel. Probs is not the theoretical upper bound for the ranking performance, the graded Rel. Probs is.
\end{remark}


\begin{table*}[t]
    \caption{NDCG@10 comparison of~\OurMethod{} and \ac{AC} with \ac{PBM} assumption on Yahoo! Webscope and MSLR-WEB30k datasets, when cascade models are used for simulating the clicks.
    }
    \label{tab:results:cmrankingcomp}
    \setlength{\tabcolsep}{0.6em}
    \begin{tabular}{l l cccc| cccc}
    \toprule
        &
        & \multicolumn{4}{c|}{Yahoo! Webscope}
        & \multicolumn{4}{|c}{MSLR-WEB30k}
        \\
        \cmidrule(lr){3-6}\cmidrule(lr){7-10} 
        &
        & \multicolumn{2}{c}{DCM} & \multicolumn{2}{c|}{UBM}
        & \multicolumn{2}{|c}{DCM} & \multicolumn{2}{c}{UBM}
        \\
        \cmidrule(lr){3-4}\cmidrule(lr){5-6}\cmidrule(lr){7-8}\cmidrule(lr){9-10} 
        &
        & \bf Binarized           & \bf Graded          
        & \bf Binarized           & \bf Graded
        & \bf Binarized           & \bf Graded          
        & \bf Binarized           & \bf Graded
        \\
    
        \midrule
        \multirow{4}{*}{LambdaMART}

& No Correction
& 0.694 & 0.691 & 0.691 & 0.690 & 0.405 & 0.402 & 0.397 & 0.400 \\
& AC
& 0.739 & 0.753 & 0.753 & 0.759 & \bfseries 0.411\rlap{$^\dagger$} & 0.469 & 0.387 & \bfseries 0.474\rlap{$^*$} \\
& \bfseries MBC
& \bfseries 0.752\rlap{$^*$} & \bfseries 0.760\rlap{$^*$} & \bfseries 0.756\rlap{$^*$} & \bfseries 0.763\rlap{$^\dagger$} & 0.407 & \bfseries 0.477\rlap{$^*$} & \bfseries 0.410\rlap{$^*$} & 0.469 \\

        \cmidrule(lr){2-10}
        \cmidrule(lr){2-10}
    
& Rel. Probs
& 0.756 & 0.770 & 0.756 & 0.770 & 0.409 & 0.485 & 0.409 & 0.485 \\
    
        \midrule
        \multirow{4}{*}{DNN}

& No Correction
& 0.709 & 0.712 & 0.709 & 0.716 & 0.407 & 0.407 & 0.391 & 0.406 \\
& AC
& 0.728 & 0.733 & 0.736 & \bfseries 0.743 & 0.364 & 0.433 & 0.402 & \bfseries 0.445\rlap{$^*$} \\
& \bfseries MBC
& \bfseries 0.741\rlap{$^*$} & \bfseries 0.743\rlap{$^*$} & \bfseries 0.741\rlap{$^*$} & 0.742 & \bfseries 0.408\rlap{$^*$} & \bfseries 0.446\rlap{$^*$} & \bfseries 0.410\rlap{$^\dagger$} & 0.441 \\

        \cmidrule(lr){2-10}
        \cmidrule(lr){2-10}
    
& Rel. Probs
& 0.743 & 0.749 & 0.743 & 0.749 & 0.412 & 0.453 & 0.412 & 0.453 \\

        \bottomrule
    \end{tabular}
    \end{table*}
\subsection{Click model mismatch}
\label{sec:results:cmranking_performance}
Next, we investigate the effect of a mismatch between a click model that generates clicks and a click model that is used for debiasing:
\begin{enumerate}
    \item[] \emph{How do~\OurMethod{} and \ac{AC} methods perform when \ac{PBM} is assumed for debiasing, whereas user clicks adhere to a different click model?}
\end{enumerate}

\noindent%
This is an important question as in reality none of the existing click models can fit user clicks perfectly and there is always a mismatch between a click model that user clicks adhere to and the one that is assumed for debiasing.

In this section, we want to examine the robustness of different correction methods in terms of their assumed click model.
Specifically, we simulate clicks based on two well-known models: \ac{DCM}~\citep{DCM2009}, that is a cascade-based click model, and \ac{UBM}~\citep{UBM2008}, that has features of both position-based and cascade-based models.
In order to have realistic experiments, we learn the parameters of these click models using the Yandex dataset,\footnote{\url{https://www.kaggle.com/c/yandex-personalized-web-search-challenge}} which contains a large amount of clicks from a production search engine, and the PyClick library.\footnote{\url{https://github.com/markovi/PyClick}}
Since the Yandex dataset has the top-$10$ results, the parameters are obtained for the top-$10$ and our experiments in this section are reported for the top-$10$ setting.\footnote{Hence, the Rel.~Probs results in this section may be different from those in Table~\ref{tab:results:rankingcomp}, where the top-$20$ is used instead.}
Then, we use the learned parameters to simulate clicks similarly to Sec.~\ref{sec:experiment:clicksimulation}, this time using \ac{DCM} and \ac{UBM} instead of \ac{PBM}.

For each case, we debias clicks using~\OurMethod{} and \ac{AC} with the \ac{PBM} click model.
Table~\ref{tab:results:cmrankingcomp} shows the ranking performance of these correction methods.

We observe that~\OurMethod{} always improves over the no correction case, while \ac{AC} fails to do so in some cases: Binarized UBM with LambdaMART and Binarized DCM with DNN.
Furthermore, in most cases, there is a gap between the corrections and the Rel.~Probs performance.
This gap is due to the mismatch between the actual and assumed click models.

Comparing~\OurMethod{} and \ac{AC}, we see that~\OurMethod{} outperforms \ac{AC} in most cases, suggesting that our~\OurMethod{} method is more robust to the click model mismatch.
This observation coincides with our expectation, since~\OurMethod{} does not rely on the parameters of click models: as long as the click probabilities over relevant and non-relevant items are separable for each position,~\OurMethod{} with \ac{PBM} works fine.
In other words, unlike \ac{AC} where the examination probability of each position is estimated by a single value, in~\OurMethod{} different items at the same position are allowed to have different examination probabilities (see Remark~\ref{rem:almostthesame}).

To further illustrate the difference, we give a toy example.
Suppose there are two sets of sessions $S_1$ and $S_2$ such that the true (hidden) examination probability of an item at position $4$ is $0.8$ for sessions in $S_1$ and $0.4$ for sessions in $S_2$ respectively.
Using Eq.~\ref{eq:epsilon}, the probabilities of click at position $4$ on a relevant item are $0.768$ and $0.384$ and on a non-relevant item are $0.13$ and $0.065$ for sessions in $S_1$ and $S_2$ respectively.
The best \ac{AC} can do is to estimate the examination probability by the average
of $\frac{0.8 + 0.4}{2} = 0.6$, leading to inaccurate relevant probabilities for items in both sets of sessions.
\OurMethod{}, instead, separates the \ac{CTR}s obtained from sessions.
Given enough sessions from each set, a clustering method similar to what we use in~\OurMethod{}, can distinguish between relevant \ac{CTR}s with mean $\in\{0.768, 0.384\}$ and non-relevant \ac{CTR}s with mean $\in\{0.13, 0.065\}$.
In other words, as long as the minimum relevant click probability is greater than the maximum non-relevant click probability,~\OurMethod{} manages to separate the two distributions.

\subsection{Efficiency of~\OurMethod{}}
\label{sec:result:efficiency}

We measured the running time of~\OurMethod{} and \ac{AC} on multiple cores of Intel(R) Xeon(R) Gold 5118 CPU @2.30GHz.
Here, we only report the time required for \emph{correcting} clicks, since the LTR part is the same in both~\OurMethod{} and \ac{AC}.
\OurMethod{}, requires around $110$ seconds to estimate the mixture distributions and correct clicks.
Each iteration of the \ac{rbEM} parameter estimation for \ac{AC} requires $370$ seconds for fitting the regression function and around $50$ additional seconds to update the bias parameters and target relevance probabilities.
The fact that at least 30 iterations are required to get a decent performance of \ac{AC} (see App.~\ref{sec:results:eminstability}) means that \ac{AC} requires a minimum of $(370+50) \cdot 30 = 12600$ seconds to correct clicks.
This means that~\OurMethod{} runs approximately $114$ times faster than \ac{AC}.

Of course, the choice of the regression function plays an important role in the above computations.
However, it is worth noting that even with a magical zero-time regression function, \ac{AC} would still require around $50 \cdot 30 = 1500$ seconds for updating the bias parameters and target relevance probabilities.
This hypothetical setting gives a lower bound of around $13$ times for the efficiency superiority of our~\OurMethod{} method over \ac{AC}.


\subsection{\OurMethod{} with different mixture distributions}
\label{sec:results:mixture}
\OurMethod{} relies on mixture models for correcting the clicks.
In this section, we will address this question:
\begin{enumerate}
    \item[] \emph{How do different assumptions for the distribution model of mixture components affect the correction quality?}
\end{enumerate}

\noindent%
Particularly, to compare different distribution shapes, we execute two variations of~\OurMethod:
\begin{enumerate}[leftmargin=*,nosep]
    \item {\bfseries \OurMethod{} (Gaussian):} 
    A Gaussian (normal) distribution is usually the default choice for modeling real world data, and due to the \acl{CLT} it is usually a safe choice.
    \item {\bfseries \OurMethod{} (Binomial):} 
    We include a Binomial distribution test, since the clicks are usually considered to have a Bernoulli distribution which makes \ac{CTR}, 
    follow a Binomial distribution.
\end{enumerate}

\setlength{\tabcolsep}{0.15em}

{\renewcommand{\arraystretch}{1}
\begin{figure}[t]
\centering
\begin{tabular}{l l c}
    \\
    & \rotatebox[origin=lt]{90}{\hspace{1.75em} \small nDCG@10}
    & \includegraphics[scale=0.4]{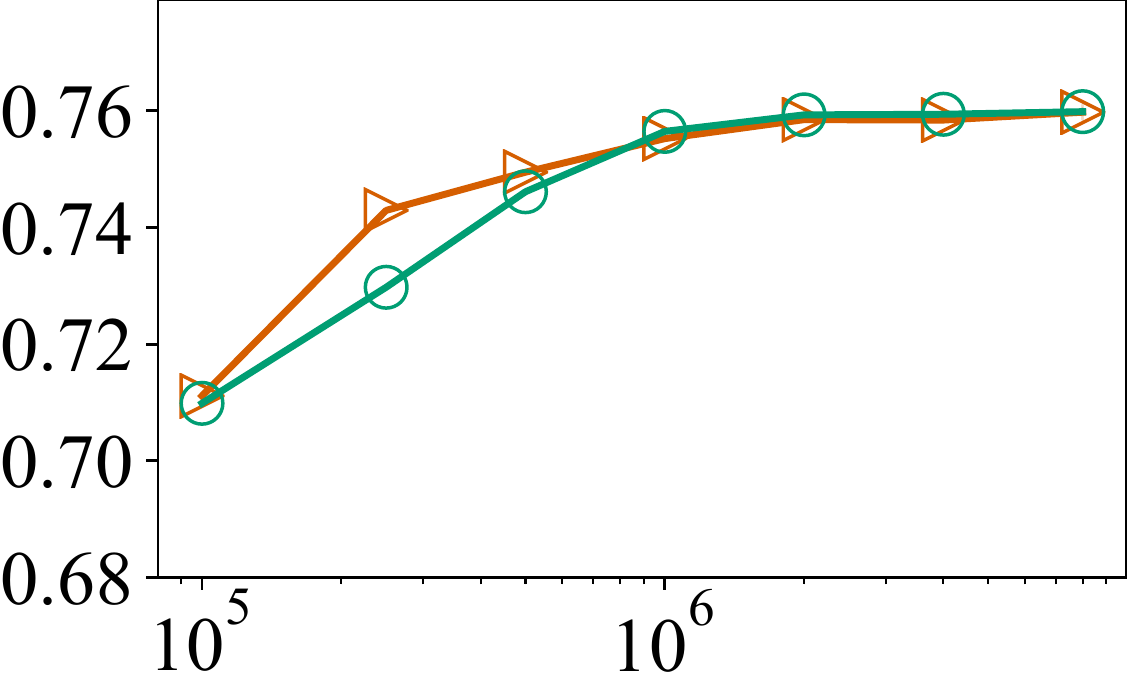}
\\ 
&
& \small Number of Training Clicks
\\ 
 \multicolumn{3}{c}{
 \includegraphics[scale=0.4]{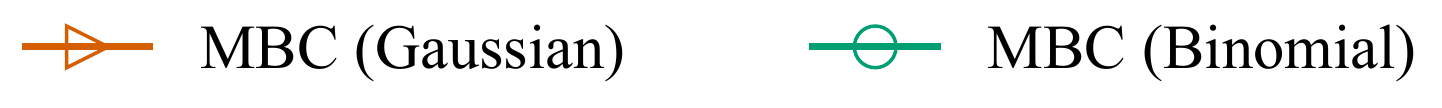}
} 
\end{tabular}
\caption{
    Ranking comparison of different mixture distributions for~\OurMethod{} in term of nDCG@10 with respect to different numbers of clicks on Yahoo! Webscope dataset.
}
\label{fig:results:mixture}
\end{figure}
}


\noindent%
Fig.~\ref{fig:results:mixture} shows the effect of the assumed distribution shape on the performance of~\OurMethod.
These experiments coincide with the theory provided in Sec.~\ref{sec:method}. 
However, we observe that the Gaussian model converges faster than Binomial model. 
We hypothesize that, since the Binomial model is less generalizable than the Gaussian model, it cannot recover the signal in the presence of high levels of noise in the low data regime.

\section{Conclusion}
We have proposed a new correction method for position and trust bias, to be used in \ac{CLTR}, and we have proven its unbiasedness.
Our method, \acl{MBC}, assumes that the distribution of \acp{CTR} of different items is a mixture of two distributions: relevant and non-relevant.
Once this mixture is estimated, the relevant items are easily identified and can be used to train a \ac{LTR} algorithm.
Consequently, correction and learning to rank are two separate phases in our~\OurMethod{} method.
This breaks the cyclic dependency between bias parameter estimation and relevance inference in existing correction methods.
Unlike those methods, in which the unbiasedness relies on accurate bias parameters estimation, the unbiasedness proof of~\OurMethod{} does not rely on knowledge of relevance.
This solves some of the practical limitations of the existing methods for correcting position and trust bias which depend on the bias parameter estimation and usually use the \ac{rbEM} for that.

Particularly, we have found that the cyclic dependency in the existing methods, leads to at least three practical limitations:
\begin{enumerate*}
    \item Severe sensitivity to the choice of the regression function;
    \item EM not necessarily converging towards the zero gradient; and,
    \item Low efficiency due to repeated use of the regression function.
\end{enumerate*}
\OurMethod{} is a new approach that solves all of these limitations as a side benefit.

We have performed extensive semi-synthetic experiments to analyze the strength of~\OurMethod{} at correcting click bias.
Our experiments show that~\OurMethod{} outperforms \ac{AC}, the state-of-the-art correction method for position and trust bias, in most of the settings.
Furthermore, we have provided evidence that~\OurMethod{} is more robust to the click model mismatch, compared to \ac{AC}.

There are several directions that can be followed for the future work:
\begin{enumerate*}
    \item Using non-parametric approaches for estimating the mixtures. 
    \item Testing on more complex click models than \ac{PBM}, as~\OurMethod{} is easier to extend to more complex click models than \ac{AC}.
\end{enumerate*}

\section*{Code and data}
To facilitate the reproducibility of the reported results, this work only made use of publicly available data and our experimental implementation is publicly available at \url{https://github.com/AliVard/MBC}.

\begin{acks}
This research was supported by Elsevier and the Netherlands Organisation for Scientific Research (NWO)
under pro\-ject nr
612.\-001.\-551.
All content represents the opinion of the authors, which is not necessarily shared or endorsed by their respective employers and/or sponsors.
\end{acks}

\appendix

\appendix
\section{Analysis of regression-based EM}
\label{app:em}
In this Appendix, we list and discuss about three specific practical limitations of \ac{rbEM} with \ac{IPS} and \ac{AC}.
As stated earlier, in \ac{AC} and \ac{IPS} there is an inherent cyclic dependency between relevance and bias parameters which leads to the use of iterative algorithms like \ac{EM}.
All of the limitations we list here come from the fact that standard \ac{EM} cannot be used to infer the bias parameters and a regression function has to be used in the middle.
Needless to say, these practical limitations do not apply to~\OurMethod.

\subsection{Practical limitations of \ac{rbEM} for \ac{AC}}
\label{app:limitations}
\paragraph{Sensitivity to the choice of regression function}
In \ac{rbEM} with \ac{AC}, the regression function is responsible for providing the relevance probabilities in the M-step.
This means that the \ac{EM} is no longer using the values obtained through the likelihood maximization, but an estimate of them obtained from the regression function.
Consequently, the performance of \ac{EM} greatly depends on the performance of the underlying regression function.
We show empirically that different choices for the regression function lead to different performances in the ranking of the final unbiased \ac{LTR} algorithm.

\paragraph{Not necessarily converging to zero gradient}
Unlike standard EM, \ac{rbEM} is not guaranteed to move in the direction of zero gradient.
The reason is simple: in the M-step of \ac{rbEM}, the relevance probabilities that maximize the likelihood function are replaced with the outputs of the regression function, in favor of addressing the otherwise unavoidable sparsity issue.
Therefore, the convergence proof of the standard EM, no longer holds for \ac{rbEM}.
Our observations suggest that when using the \ac{rbEM}, more iterations does not necessarily mean better performance, as opposed to the standard EM.

\paragraph{Low efficiency}
The \ac{rbEM}, by design, requires a regression function to be fitted to the relevance probabilities at each maximization step of the EM algorithm.
We have discussed this issue in Sec.~\ref{sec:result:efficiency}.

\subsection{Instability of \ac{rbEM} for \ac{AC}}
\label{sec:results:eminstability}
In practice, \ac{rbEM} is used  to estimate the bias parameters for \ac{AC}. In this set of experiments, we address the following two questions about the performance of \ac{rbEM} for \ac{AC}:
\begin{enumerate}[leftmargin=*,nosep]
    \item \emph{What is the impact of the choice of regression function for \ac{rbEM} on the ranking performance of \ac{AC}?}
    \item \emph{How does the ranking performance of \ac{rbEM} \ac{AC} vary as a function of the number of iterations?}
\end{enumerate}

\noindent%
Neither of the above questions concern the standard \ac{EM}, as discussed before.
Introducing a regression function in an \ac{EM} algorithm is a powerful idea to solve the issues of standard \ac{EM} in \ac{CLTR}, but it also brings its concerns as well.

We try different regression functions with different loss functions and report the ranking performance on different iterations of \ac{rbEM}.
We use the following regression functions:
\begin{enumerate*}
    \item LambdaMART; and
    \item a neural network similar to that of~\citep{vardasbi2020when}.
\end{enumerate*}
Since the regression function is used for fitting to relevance \emph{probabilities}, we use the cross entropy loss.
Following the literature, we test two cross entropy variations:
\begin{enumerate*}
    \item Sigmoid cross entropy, similar to~\citep{wang2018position,agarwal2019addressing};
    \item Soft-min-max cross entropy, similar to~\citep{vardasbi2020when}.
\end{enumerate*}


\setlength{\tabcolsep}{0.05em}

{\renewcommand{\arraystretch}{0.2}
\begin{figure}[t]
\centering
\begin{tabular}{l c c}
& \hspace{1em}\small Yahoo! Webscope
& \hspace{1em}\small MSLR-WEB30k 
\\
\rotatebox[origin=lt]{90}{\hspace{1.5em} \small nDCG@10}
& 
\includegraphics[scale=0.35]{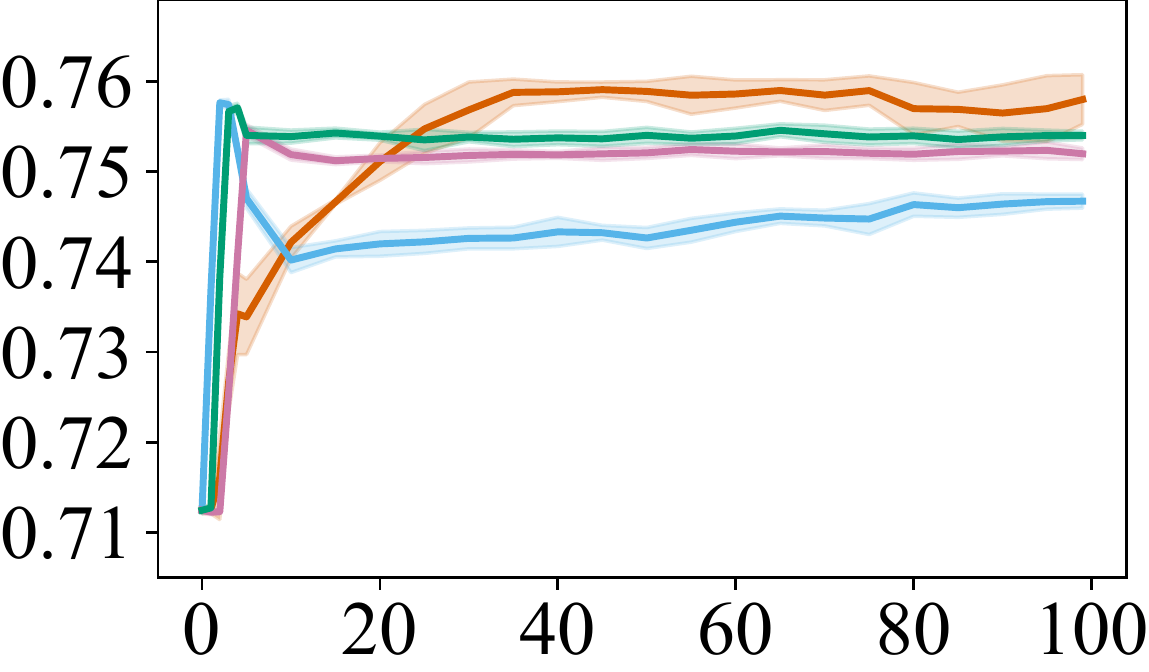} 
&
\includegraphics[scale=0.35]{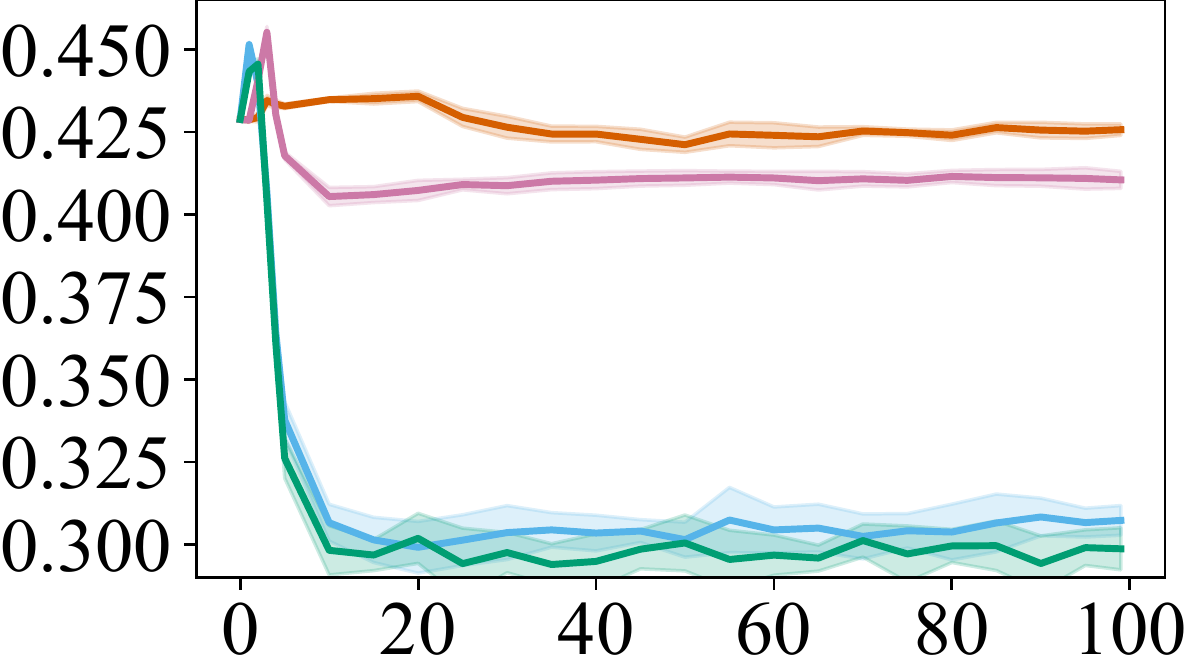}
\\
& \hspace{1em}\small EM iteration
& \hspace{1em}\small EM iteration 
\\ 
 \multicolumn{3}{c}{
 \includegraphics[scale=0.4]{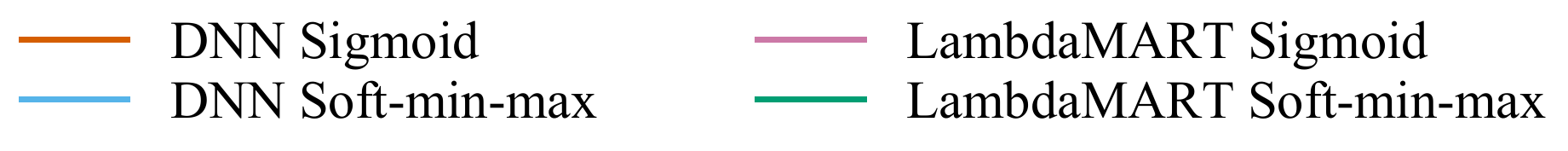}
} 
\end{tabular}
\caption{
Regression-based EM ranking performance with different regression functions with respect to EM iterations.
Left: Yahoo! Webscope dataset; right: MSLR-WEB30k dataset.
}
\label{fig:results:em_iterations}
\end{figure}

}

Fig.~\ref{fig:results:em_iterations} summarizes the ranking performance of \ac{AC} with \ac{rbEM}, with different regression functions, as a function of \ac{EM} iterations.
Based on the observations in this figure, the answers to the questions of this section are: \emph{big} and \emph{a lot}.
Concerning the first question, we see that different regression functions lead to large differences in ranking performance.
More interestingly, the ordering of the regression functions is not preserved in different datasets.

The second question is about the change of the performance as the number of EM iterations increases.
Fig.~\ref{fig:results:em_iterations} shows that the ranking performance does not always improve with more EM iterations.
On the Yahoo! Webscope dataset, the DNN with sigmoid loss has a slight performance drop at iteration 80.
On the MSLR-WEB30k dataset, the performance of DNN with sigmoid loss is decreasing between iterations 20 and 50.
Another observation relating to the EM iterations are the sudden performance drops at some iterations: DNN with Soft-min-max loss in both datasets.
These observations indicate that, unlike the standard EM, the rbEM cannot necessarily be trusted with regard to iterations:
The performance is not always increasing with the number of iterations.

Based on the above discussions and according to Fig.~\ref{fig:results:em_iterations}, our choice for the rbEM baseline for comparison with other methods is as follows:
We chose the DNN with sigmoid cross entropy loss function as the regression function.
When there are no anomalies, the DNN Sigmoid performs well up until iteration 100.
In the cases where there is an anomaly, we use the results of the last iteration before the anomaly.

\bibliographystyle{ACM-Reference-Format}
\bibliography{references}

\end{document}